\documentclass[format=acmsmall, review=false]{acmart}

\usepackage{acm-ec-19}

\usepackage{mathtools}
\usepackage{algorithmic}
\usepackage[linesnumbered,lined,boxed,commentsnumbered]{algorithm2e}
\SetKwComment{Comment}{$\triangleright$\ }{}
\usepackage{enumitem}

\usepackage{wrapfig}
\newcommand{\AutoAdjust}[3]{\mathchoice{ \left #1 #2  \right #3}{#1 #2 #3}{#1 #2 #3}{#1 #2 #3} }
\newcommand{\Xcomment}[1]{{}}

\newcommand{\InParentheses}[1]{\AutoAdjust{(}{#1}{)}}

\newcommand{\be}{\begin{equation}}
\newcommand{\ee}{\end{equation}}
\newcommand{\argmin}{\mathop{\rm argmin}}
\newcommand{\argmax}{\mathop{\rm argmax}}

\newcommand{\bydef}{\stackrel{\bigtriangleup}{=}}

\newcommand{\eqdef}{\stackrel{\textrm{def}}{=}}
\newcommand{\eps}{\varepsilon}

\newcommand{\vect}[1]{\ensuremath{\mathbf{#1}}}
\newcommand{\R}{\mathbb{R}}

\newcommand{\agents}{\mathcal{N}}
\newcommand{\items}{\mathcal{M}}

\newcommand{\allocations}{\mathcal{X}}

\newcommand{\nashw}{\texttt{NW}}
\newcommand{\optnw}{\texttt{opt}}

\newcommand{\val}{v}

\newcommand{\vali}[1][i]{{\val_{#1}}}

\newcommand{\alloc}{X}
\newcommand{\allocs}{{\mathbf \alloc}}

\newcommand{\alloci}[1][i]{{\alloc_{#1}}}
\newcommand{\yalloc}{Y}
\newcommand{\yallocs}{{\mathbf\yalloc}}

\newcommand{\yalloci}[1][i]{{\yalloc_{#1}}}
\newcommand{\zalloc}{Z}
\newcommand{\zallocs}{{\mathbf\zalloc}}

\newcommand{\zalloci}[1][i]{{\zalloc_{#1}}}
\newcommand{\halloc}{\widehat{X}}
\newcommand{\hallocs}{\widehat{\mathbf \alloc}}

\newcommand{\halloci}[1][i]{{\halloc_{#1}}}

\newcommand{\sallocs}{{\mathbf \alloc}^{*}}

\newcommand{\salloci}[1][i]{{\alloc_{#1}^*}}

\newcommand{\bundle}{B}
\newcommand{\bundlei}[1][i]{{\bundle_{#1}}}
\newcommand{\bundles}{\vect{\bundle}}

\newcommand{\efxgraph}{\ensuremath{\texttt{EFX-feasibility-graph}}}
\newcommand{\EFXgraph}[1]{\ensuremath{\texttt{EFX-feasibility-graph}\InParentheses{#1}}}
\newcommand{\robustd}[1]{\texttt{robust-demand}\InParentheses{#1}}

\newcommand{\augmentingpath}{\ensuremath{\texttt{augmenting-path}}}
\newcommand{\AugmentingPath}[1]{\ensuremath{\texttt{augmenting-path}\InParentheses{#1}}}

\newcommand{\true}{\ensuremath{\texttt{true}}}
\newcommand{\false}{\ensuremath{\texttt{false}}}
\newcommand{\itrem}{\ensuremath{\texttt{itemremoved}}}

\newcommand{\Ocomplex}[1]{O\left(#1\right)}

\definecolor{WildStrawberry}{RGB}{255,67,164}

\newtheorem{definition}{Definition}
\newtheorem{lemma}{Lemma}
\newtheorem{claim}{Claim}

\newtheorem{theorem}{Theorem}

\theoremstyle{definition}

\theoremstyle{definition}

\title{Envy-freeness up to any item with high Nash welfare:\\The virtue of donating items}

\author{Ioannis Caragiannis}
\affiliation{%
	\institution{University of Patras \& CTI ``Diophantus''}}
\author{Nick Gravin}
\affiliation{%
	\institution{Shangai University of Finance and Economics, ITCS}}
\author{Xin Huang}
\affiliation{%
	\institution{The Chinese University of Hong Kong}}

\date{} 

\begin{abstract}
Several fairness concepts have been proposed recently in attempts to approximate envy-freeness in settings with indivisible goods. Among them, the concept of envy-freeness up to any item (EFX) is arguably the closest to envy-freeness. Unfortunately, EFX allocations are not known to exist except in a few special cases. We make significant progress in this direction. We show that for every instance with additive valuations, there is an EFX allocation of a subset of items with a Nash welfare that is at least half of the maximum possible Nash welfare for the original set of items. That is, after donating some items to a charity, one can distribute the remaining items in a fair way with high efficiency. This bound is proved to be best possible.
Our proof is constructive and highlights the importance of maximum Nash welfare allocation. Starting with such an allocation, our algorithm decides which items to donate and redistributes the initial bundles to the agents, eventually obtaining an allocation with the claimed efficiency guarantee. The application of our algorithm to large markets, where the valuations of an agent for every item is relatively small, yields EFX with almost optimal Nash welfare. To the best of our knowledge, this is the first use of large market assumptions in the fair division literature. We also show that our algorithm can be modified to compute, in polynomial-time, EFX allocations that approximate optimal Nash welfare within a factor of at most $2\rho$, using a $\rho$-approximate allocation on input instead of the maximum Nash welfare one.
\end{abstract}

\begin{document}

\maketitle 

\section{Introduction}
\label{sec:intro}
The agenda of fairly allocating items to agents has received much attention by the EconCS community recently. Most of the solution concepts with compelling properties such as envy-freeness and proportionality cannot always be attained when the items are indivisible, in contrast to the traditionally studied case of divisible items. Recent work put forth several  promising fairness notions that adapt and ``approximate'' the definition of envy-freeness to the case of individual items and demonstrated that allocations that satisfy some of these notions are possible, with some minimum efficiency guarantees at the same time.


A very natural adaptation of envy-freeness for settings with indivisible items was defined by Budish~\cite{budish2011combinatorial} (but is also implicit in the earlier work of Lipton et al.~\cite{lipton2004approximately}). An allocation of items is envy-free up to some item (EF1) if any possible envy of an agent for the allocation of another can be eliminated by removing {\em some} item from the envied bundle. It is well known that EF1 allocations can always be achieved, and recently Caragiannis et al.~\cite{caragiannis2016unreasonable} showed that the Nash welfare maximizing allocation is Pareto-optimal and EF1. This is indeed a success story for the Nash welfare efficiency measure defined as the geometric mean of the agent valuations, a standard and well studied objective that is particularly well aligned with fair division applications.


\begin{wrapfigure}[6]{r}{0.26\textwidth}
	\centerline{
		\begin{tabular}{l|c c c c}
			& c & r & p & n \\\hline
			Alice & $10$ & $9$ & $4$ & $6$ \\
			Bob & $10$ & $6$ & $9$ & $4$ \\
			Carol & $10$ & $4$ & $6$ & $9$
	\end{tabular}}
\end{wrapfigure}
For example, consider the problem of fairly dividing the inheritance of four items, a car, a ring, a painting, and a necklace, among three agents named Alice, Bob, and Carol. Agents have preferences for the items, which are expressed in the table at the right, where each number represents the amount in thousands of USD of how much each agent values the corresponding item. Clearly, there is no envy-free allocation in this example, because someone must get at least two items, and then one of the two other agents would have a higher total value for these items compared to what the agent got. The allocation which gives the ring to Alice, the car and the painting to Bob, and the necklace to Carol, is EF1 with the valuations of $9$, $19$, and $9$, respectively. In this allocation, Bob does not envy anyone, while Alice and Carol would prefer to swap Bob's allocation with their own: Alice values the car and the painting in Bob's bundle for the total of $14$, and Carol values them for $16$. This is still an EF1 allocation as Alice's and Carol's envy can be eliminated by ignoring their most valued item in Bob's bundle (i.e., the car) when comparing valuations.

The above allocation in our inheritance example maximizes the Nash welfare objective, which is widely considered as an important efficiency measure for allocations. It provides an appealing compromise between utilitarian and egalitarian social welfare objectives, defined as the sum of agent valuations for their allocated items and minimum valuation among all agents, respectively. Finding a Nash welfare maximizing allocation is well-known to be computationally tractable in settings with divisible items (by minimizing the convex program of Eisenberg and Gale~\cite{EG59}), where the resulting allocation also happens to be envy-free. In the case of indivisible goods, besides being EF1, Nash welfare maximizing allocations have some additional fairness properties (see e.g.~\cite{caragiannis2016unreasonable}), including a remarkable {\em scale-invariant} property that rescaling the valuation of an agent does not change the allocation that maximizes Nash welfare (see~\cite{Moulin03} for a detailed discussion).
%


On the other hand, the Nash welfare maximizing allocation from the example above can hardly be considered fair, since Bob gets more value than Alice and Carol put together in our symmetric instance\footnote{The instance is invariant under a cyclic permutation of agents' names and the respective cyclic permutation of all items except the car.}. Consider another allocation, where Alice gets the ring for the value of $9$, Bob gets the car for the value of $10$, while the painting and the necklace go to Carol with the total value of $15$. This allocation is slightly sub-optimal in terms of the Nash welfare, but it looks clearly better than the previous allocation from the fairness point of view. Indeed, observe that, inevitably, one agent must get the car; it is Bob in our case. This already creates some envy from Alice and/or Carol. So, it makes more sense to allocate the remaining three items only to them and give nothing more to Bob. More concretely, such allocations have the much more appealing property of envy-freeness up to {\em any} item (EFX), introduced by Caragiannis et al. \cite{caragiannis2016unreasonable}. Now, the envy of Alice and Bob for Carol can be eliminated by ignoring their {\em least-valued} item in Bob's bundle when comparing valuations. 
Arguably, EFX is the best fairness analog of envy-freeness for indivisible items. But even though EF1 allocations can be computed in many different ways (e.g., by a simple draft mechanism, by the local-search algorithm of Lipton et al.~\cite{lipton2004approximately}, by maximizing the Nash welfare, etc.), for EFX allocations, we do not even know whether they always exist. Moreover, asking for EFX allocations that simultaneously satisfy some efficiency guarantee is way beyond the reach of known techniques. 


\paragraph{Donation of items.}
A natural way to deal with the (hypothetical) issue of the existence of EFX allocations is to consider partial allocations, which leave some of the items unallocated. The unallocated items could then be donated to a charity or sold at auction with the profit equally distributed between the agents. Actually, in situations like the above inheritance example that arise in practice, it is usually observed that people donate some of their possessions to charities before distributing the rest to their heirs. Undoubtedly, the donation of items can completely eliminate envy, e.g., by donating all items and allocating nothing to the agents. However, although beneficial for the charity cause, this outcome cannot be considered as good for the problem at hand. Therefore, an important question to ask is: 
\begin{quote}
How can we achieve an EFX allocation of high Nash welfare by donating items smartly?
\end{quote}


\paragraph{Our results.}
We make careful use of the idea of donating items and hit many birds with one stone. We present an algorithm that takes an allocation of maximum Nash welfare as input and, after donating some of the items, outputs an allocation of the remaining ones that is EFX and has a Nash welfare that is at least half the optimal one for the original instance (that includes the donated items). As we show, this efficiency guarantee is best possible for EFX allocations, no matter whether there are donated items or not. A remarkable bonus property is that the bundle that is allocated to each agent is a subset of the bundle he/she initially gets in the maximum Nash welfare allocation. 
When the same approach is applied to {\em large market} instances, where any item's contribution is only a small fraction of the total value of any agent, our algorithm outputs almost optimal allocations. 
The algorithm can be modified to work with input allocations that have sub-optimal Nash welfare. Combined with a $\rho$-approximation algorithm for maximizing Nash welfare (such as the algorithms of Cole and Gkatzelis~\cite{cole2015approximating} or Barman et al.~\cite{barman2018finding}), the modified algorithm runs in polynomial-time and computes an EFX allocation that is $2\rho$-approximation to the maximum Nash welfare.

\paragraph{Techniques.}
We remark that our algorithms are purely combinatorial. This comes in contrast to the heavily used convex programming and rounding techniques for Nash welfare maximization as well as the use of duality techniques and item prices for the computation of fair allocations. 

Our algorithms do not reallocate any items between the bundles in the initial allocation but permanently remove items from the instance. This removal operation gives us fine control of the agents' preferences. For example, if an agent is ``happy'' with a bundle from which we do not remove an item, then she stays happy with it after the removal of an item. At a high level, we always seek to find a complete EFX feasible matching of agents with bundles and carefully remove an item so that the size of the matching never decreases. From the Nash welfare maximality property, we derive that no bundle can lose more than a certain fraction of its value for the agent who had it in the initial Nash welfare maximizing allocation, yielding the desired Nash welfare guarantee. 
For large markets, we are able to obtain our considerably improved result by proving a number-theoretic lemma 
using Karamata's inequality for the logarithm function.

\paragraph{Significance.}
We employ the natural and practical operation of donating/removing items to resolve the challenging issues related to the EFX solution concept. Our approach contrasts with the recent attempts to mitigate such challenges by considering approximate versions of EFX (e.g., in the papers by Plaut and Roughgarden~\cite{plaut2018almost} and Amanatidis et al.~\cite{ABM18}). Instead, we keep the precise definition of EFX and use approximation only for the efficiency guarantee which is unavoidable. This approach of considering item donations could be useful in the study of the interplay of other fairness notions with efficiency. Furthermore, to the best of our knowledge, this is the first time that large market assumptions are considered in the fair division literature.

\subsection{Related work}
\label{sec:related}
\paragraph{Fair division with indivisible items.} The rigorous study of fair division with divisible items has a long history; it origininates from the work of Steinhaus~\cite{steinhaus1948problem} in the 40's and includes very recent breakthroughs like the envy-free cake-cutting protocol of Aziz and McKenzie~\cite{AM16}. In contrast, fairly allocating indivisible items among agents has not been as popular until very recently. Most probably, the reason for this delay is that the beautiful fairness notions of envy-freeness and proportionality, that have received so much attention in the literature on divisible items, are rarely achievable with indivisible items. To give an embarassing example, just consider a single item and two agents with identical valuations for the item.

The recent interest for the indivisible items setting was sparked with the definition of fairness notions that approximate envy-freeness and proportionality. In particular, the notion of maxmin fair share (MMS), defined by Budish~\cite{budish2011combinatorial} and Bouveret and Lema\^{i}tre~\cite{BL16} can be thought of as an approximate version of proportionality and has received much attention recently, e.g., see~\cite{procaccia2014fair,AMNS17,barman2017approximation,ghodsi2018fair}.
Besides the concepts of EF1 and EFX mentioned above, approximate versions of envy-freeness include epistemic envy-freeness~\cite{ABCGL18} or notions that require the minimization of the envy-ratio~\cite{lipton2004approximately} and degree of envy~\cite{CEEM07,NR14} objectives.

As mentioned earlier, EF1 is easy to achieve with several different methods. What is really challenging is to achieve EF1 and Pareto-optimality simultaneously. This was proved to be possible in \cite{caragiannis2016unreasonable}, by an allocation of maximum Nash welfare. The popular website Spliddit \cite{GP14}, available at \url{www.spliddit.org}, returns such allocations as part of its ``Divide goods'' application. Following \cite{caragiannis2016unreasonable}, Barman et al.~\cite{barman2018finding} investigate whether EF1 and Pareto-optimal allocations can be computed efficiently and present a pseudo-polynomial-time algorithm.

Unlike EF1, the existence of EFX allocations is still a mystery, even for three agents with additive valuations. Plaut and Roughgarden \cite{plaut2018almost} prove the existence of EFX allocations for setting with two agents only or with more agents and identical valuations. In addition, they present an algorithm for computing an $1/2$-EFX allocation, where the value of each agent from her bundle is at least half of what the EFX property requires it to be. Caragiannis et al.~\cite{caragiannis2016unreasonable} and Amanatidis et al.~\cite{ABM18} consider different approximations of EFX (namely, approximations of the strongly related pairwise MMS fairness concept). 

\paragraph{Nash welfare.} The history of Nash social welfare (or, simply, Nash welfare) dates back to the 50's, where it was used in bargaining problems~\cite{nash1950bargaining}. In fair division, it is considered as a good measure to balance fairness with efficiency. In particular, in settings with divisible items, maximizing the Nash welfare can be done in polynomial-time using the convex program of Eisenberg and Gale~\cite{EG59}, which also leads to envy-free allocations. Unfortunately, such nice properties dissappear in the indivisible items setting as the problem of computing an allocation of maximum Nash welfare becomes APX-hard~\cite{lee2017apx} (see also \cite{NNRR14} for a weaker hardness result). Still, as the Nash welfare is a very important efficiency measure, several constant-approximation algorithms have been proposed recently. The first such algorithm was due to Gkatzelis and Cole~\cite{cole2015approximating} and approximates the Nash welfare within a factor of $2.89$. This was improved to $2$ by Cole et al.~\cite{cole2017convex}, and further to $1.45$ by  Barman et al.~\cite{barman2018finding}.  See also~\cite{anari2018nash,anari2017nash,garg2018approximating} for approximation algorithms in more general settings (with non-additive valuations). With the exception of \cite{barman2018finding} that uses item pricing techniques, rounding of convex programming relaxations is the main algorithmic tool in this line of research.

\paragraph{Other related papers.} Our main positive result bounds the efficiency gap of the best EFX allocation in terms of Nash welfare. As such, this is a {\em price of fairness} result, a notion that was introduced independently by Bertsimas et al.~\cite{bertsimas2011price} and Caragiannis et al.~\cite{caragiannis2012efficiency}. Also, the general structure of our algorithms has conceptual similarities (but also many technical differences) to algorithms for computing combinatorial walrasian equilibria~\cite{FGL16}, where items are packaged into groups instead of being donated. Finally, we remark that large market assumptions have been considered in a few different areas recently, such as in the study of the price of anarchy of large games~\cite{cole2016large,feldman2016price}, in budget-feasible mechanisms~\cite{AGN14}, and in optimizing over (very different from ours) fairness constraints~\cite{BS15}.

\subsection{Roadmap}
The rest of the paper is structured as follows. We begin with formal definitions in Section~\ref{sec:prelim}. Our counter-example showing that the Nash welfare of EFX allocations can be far from optimal is given in Section~\ref{sec:example_lb}. The description and analysis of our main positive result is presented in Section~\ref{sec:algorithm}. Our modifications that allow the algorithm work with allocations of sub-optimal Nash welfare on input are given in Section~\ref{sec:poly-time-algorithm}. The large market assumption is considered in Section~\ref{sec:large_market}. We conclude in Section~\ref{sec:open} with a short discussion and open problems.

\section{Preliminaries}
\label{sec:prelim}

Let $\agents$ be the set of $n$ agents and $\items$ be the set of $m$ indivisible items. We enumerate agents and items from $1$ through $n$ and through $m$, respectively. We will often refer to the set of agents $\agents$ as $[n]$ and to the set of items $\items$ as $[m]$, where $[k]\eqdef\{1,2,\ldots,k\}$ for any positive integer number $k$. Each agent $i\in\agents$ is endowed with a valuation function $\vali:\items\to \R_{+}$, where $\R_+$ is the set of positive real numbers. Valuations are additive; by slight abuse of notation, agent $i$ has a value of $\vali(S)\eqdef\sum_{g\in S}\vali(g)$ for any set $S\subseteq\items$. 

We refer to a set of items as a {\em bundle} and extensively consider partitions of items into bundles. Formally, a {\em partition} of a set of items $S\subseteq\items$ is a set of $n$ disjoint bundles $\bundles=\left\{\bundlei[1],\ldots,\bundlei[n]\right\}$ such that $\bundlei\cap\bundlei[j]=\emptyset$ and $\bigcup_{i\in[n]} \bundlei=S$. 
An {\em allocation} $\allocs=\left(\alloci[1],\ldots,\alloci[n]\right)$ of the items in $S\subseteq\items$ to the agents of $\agents$ is an ordered partition of the item set $S$, where each agent $i\in[n]$ gets the bundle $\alloci$. 
We denote the set of all possible allocations of $S$ to the agents of $\agents$ as $\allocations(S)$. 
From now on, the term allocation implies an allocation to the agents of $\agents$.

An allocation $\allocs=\left(\alloci[1],\ldots,\alloci[n]\right)$ is called
\begin{itemize}[noitemsep,nolistsep]
\item {\em envy-free} (\textbf{EF}) if $\vali(\alloci)\ge\vali(\alloci[j])$ for any $i,j\in\agents$; 
\item {\em envy-free up to one item} (\textbf{EF1}) if $\vali(\alloci)\ge\min_{g\in\alloci[j]}\vali(\alloci[j]-\{g\})$ for any $i,j\in\agents$; 
\item {\em envy-free up to any item} (\textbf{EFX}) if $\vali(\alloci)\ge\max_{g\in\alloci[j]}\vali(\alloci[j]-\{g\})$ for any $i,j\in\agents$\footnote{If the set $\alloci[j]$ is empty we assume that $\min_{g\in\alloci[j]}\vali(\alloci[j]-\{g\})=\max_{g\in\alloci[j]}\vali(\alloci[j]-\{g\})=0$. We remark that our definition is slightly different than the one in \cite{caragiannis2016unreasonable}, where the maximum is taken over the items of bundle $\alloci[j]$ for which agent $i$ has strictly positive valuation. However, our positive results are only stronger in this way, and in the proof of Theorem~\ref{thm:lower-bound}, all agent valuations are non-zero.}.
\end{itemize}

We measure the efficiency of an allocation $\allocs=(\alloci[1],\ldots,\alloci[n])$ using its {\em Nash social welfare} (or, simply {\em Nash welfare}) $\nashw(\allocs)$, which is defined as the geometric mean of $\{\vali(\alloci)\}_{i=1}^{n}$. For any set of items $S\subseteq\items$, we denote by $\optnw(S)$ the maximum Nash welfare over all allocations of $S$. Formally,
\[
\nashw(\allocs)\eqdef\left(\prod_{i\in[n]}~\vali(\alloci)\right)^{1/n},\quad\quad\quad
\optnw(S)\eqdef\max_{\allocs\in\allocations(S)}\nashw(\allocs).
\]
Our goal is to find an EFX allocation with as large Nash welfare as possible. We say that the allocation $\allocs$ of $S\subseteq\items$ is {\em $\alpha$-efficient} if $\alpha\cdot\nashw(\allocs)\ge\optnw(\items)$. An allocation $\allocs$ is called {\em Pareto-optimal} if there is no other allocation $\yallocs$ with at least as high value for every agent $i$, i.e., $\vali(\yalloci)\ge\vali(\alloci)$, and strictly higher value for one of the agents. In particular, the Nash welfare maximizing allocation is Pareto-optimal.

We now introduce {\em large markets}. These are allocation problems, in which the valuation of every agent for any item is a small fraction of her total value for all items of $\items$.

\begin{definition}[Large Market]
\label{def:large_market}
An allocation problem satisfies the large market condition with parameter $\eps\in (0,1]$ if $\vali(g)\le\frac{\eps}{n}\vali(\items)$ for every agent $i\in[n]$ and any item $g\in\items$.
\end{definition}
In Section~\ref{sec:large_market}, we actually use the following weaker large market condition, which is defined with respect to any allocation of optimal Nash welfare.
\begin{definition}[Large market w.r.t. optimal allocation $\allocs$]
\label{def:large_market_allocation}
An allocation $\allocs$ of optimal Nash welfare satisfies the large market condition with parameter $\eps>0$ if $\vali(g)\le \eps\cdot\vali(\alloci)$ for every agent $i\in[n]$ and any item $g\in\alloci$.
\end{definition}

\begin{claim}
\label{cl:large_market_equivalence}
The large market condition with parameter $\eps$ implies the large market condition for an optimal Nash welfare allocation with parameter $\eps'=\eps/(1-\frac{n-1}{n}\eps)$.
\end{claim}

\begin{proof}
We use the fact that a Nash welfare maximizing allocation $\allocs$ is also EF1 \cite{caragiannis2016unreasonable}. Let us fix any agent $i\in\agents$. The EF1 condition for agent $i$ implies that $\forall j\neq i$, $\exists g_{ij}\in\alloci[j]$ such that $\vali(\alloci)\ge\vali(\alloci[j]\setminus\{g_{ij}\})$. By summing these inequalities over all $j\neq i$ and the equality $\vali(\alloci)=\vali(\alloci)$, we get
\[
n\cdot\vali(\alloci)\ge\vali(\alloci)+\sum_{j\neq i}\vali(\alloci[j]\setminus\{g_{ij}\})=\vali(\items)-\sum_{j\neq i}\vali(g_{ij})\ge\vali(\items)-(n-1)\cdot\frac{\eps}{n}\vali(\items).
\] 
Therefore, for any agent $i$ and item $g\in\alloci$ we have
\[
\vali(g)\le\frac{\eps}{n}\vali(\items)\le \frac{\eps}{n} n\cdot\vali(\alloci)\left(1-\frac{n-1}{n}\eps\right)^{-1}=\eps'\cdot\vali(\alloci).\qedhere
\]
\end{proof}



\section{A lower bound on the Nash welfare of EFX allocations}
\label{sec:example_lb}
Before presenting our positive results, we give an allocation problem where almost half of the optimal Nash welfare must be sacrificed in order to achieve EFX, either for the original set of items or for any subset of them.


 
\begin{theorem}\label{thm:lower-bound}
For any positive integer $n$ and $\eps>0$, there is an allocation problem with $n$ agents and set of items $\items$, such that any EFX allocation $\allocs$ of any subset of items $S\subseteq\items$ satisfies $\left(2^{1-1/n}-2\eps\right) \cdot\nashw(\allocs) \leq \optnw(\items)$.
\end{theorem}

\begin{proof} Consider the allocation problem with $n$ agents, set of items  $\items\eqdef\{1,...,2n-1\}$, and valuation function $\vali:\items\to\R_+$ of agent $i\in[n]$ that is defined as follows:
\[
 \vali(g) \eqdef 
\begin{cases}
1  & \text{when } g\in\{1,...,n-1\},\\
1-\eps & \text{when } g=2n-i,\\
^\eps\!/_{2n}, & \text{ otherwise.}
\end{cases}
\]
The allocation $\sallocs=(\salloci[1], ..., \salloci[n])$ with $\salloci=\{i, 2n-i\}$ for $i\in[n-1]$, and $\salloci[n]=\{n\}$ has Nash welfare $\nashw(\sallocs)=\sqrt[n]{(2-\eps)^{n-1}(1-\eps)}$. 
Hence, the optimal Nash welfare is 
\begin{align*}
\optnw(\items)&\geq \sqrt[n]{(2-\eps)^{n-1}(1-\eps)}=2^{1-1/n}\cdot\left(1-\frac{\eps}{2}\right)^{(n-1)/n}(1-\eps)^{1/n}\\
&> 2^{1-1/n}\left(1-\frac{\eps}{2}\cdot\frac{n-1}{n}\right)\left(1-\frac{\eps}{n}\right)>2^{1-1/n}-2\eps.
\end{align*}
The second inequality follows by Bernoulli inequality $(1-x)^r> 1- xr$ for $0<r<1$. 

Now consider any EFX allocation $\allocs=(\alloci[1],...,\alloci[n])$ of a set of items $S\subseteq \items$. By the pigeonhole principle, there must be some agent $j$ who does not receive any of the first $n-1$ items. Then, agent $j$'s value is at most $\vali[j](\alloci[j])\le\sum_{g=n}^{2n-1}\vali(g)<1-\frac{\eps}{2}.$ It means that, if any other agent $i\neq j$ receives a large item $g\in[n-1]$, then her allocation must contain only that item, i.e., $\alloci=\{g\}$. Indeed, if there is an item $g'$ different than $g$ so that $g$ and $g'$ belong to bundle $\alloci$, then the EFX condition is violated, as $\vali[j](\alloci[j])<1-\eps/2<1\le\vali[j](\alloci\setminus g')$. This means that no agent gets value higher than $1$: either agent $i$ gets a large item $g\in[n-1]$ and, subsequently, $\vali(\alloci)=\vali(g)=1$, or she gets (a subset of) items $n, n+1, ..., 2n-1$ and has value $\vali(\alloci)\le\sum_{g=n}^{2n-1}\vali(g)<1-\eps/2$. Therefore, $\nashw(\allocs)\le 1$ which, together with the inequality on $\optnw(\items)$ above, completes the proof.
\end{proof}

\section{Main algorithm}
\label{sec:algorithm}
We now present our algorithm for computing an EFX allocation of some of the items to the agents. Together with the allocation problem (a set of items $\items$ and $n$ agents with valuations for the items in $\items$), the algorithm receives as input a Nash welfare maximizing allocation $\allocs$ for it. It tries to match as many bundles from the initial allocation to agents as possible and repeatedly removes items from the bundles as long as this matching does not correspond to an EFX allocation. The algorithm is guaranteed to output an EFX allocation $\yallocs$ with Nash welfare which is at least half of the original one. More precisely, we will show that $2^{\frac{n-1}{n}}\cdot\nashw(\yallocs)\ge\nashw(\allocs)=\optnw(\items)$.
%

Before presenting our algorithm in detail, we introduce several useful notions, giving forward pointers to the lines of the pseudocode (see Algorithm~\ref{alg:1}) where these notions are used. During its execution, the algorithm maintains a disjoint set $\zallocs$ of $n$ bundles. Initially (see Line 1), $\zallocs$ consists of the bundles in $\allocs$. The algorithm progresses in rounds. In each round (defined by an execution of Lines 3-12 of the pseudocode in the ``repeat-until'' loop), it tries to compute an EFX allocation (with particular properties) by assigning the bundles of $\zallocs$ to the agents. Whenever this is not possible, it removes an item from a bundle of $\zallocs$ and proceeds to the next round. When a bundle of $\zallocs$ misses an item, it becomes a {\em touched} bundle. 

In order to compute the EFX allocation in each round, the algorithm uses the notion of the {\em EFX feasibility graph}, which is defined as follows. The EFX feasibility graph is a bipartite graph $G=(\agents \cup \zallocs,E)$ between two sets of vertices: the first part contains all agents $\agents$, each agent as a vertex; the second part contains the bundles in $\zallocs$, each bundle $\zalloci$ as a vertex. Edges of $G$ are defined as
\[
E(G)\eqdef\{(i,\zalloci[j]) | \text{ (i) } \zalloci[j] \text{ is EFX feasible for } i; \text{ (ii) } \vali(\zalloci[j]) > \vali(\zalloci) \text{ if } i\neq j\}.
\]
In particular, EFX feasibility in condition (i) requires that $\vali(\zalloci[j])\ge\max_{g\in\zalloci[k]}\vali(\zalloci[k]-\{g\})$ for all $k\in[n]$. Condition (ii) in the definition of $E(G)$ restricts the space of possible allocations and expresses our preference for matching agents to their initial bundles whenever possible. The call of the subroutine \efxgraph\ in Line 3 of Algorithm~\ref{alg:1} builds the EFX feasibility graph for the current set of bundles $\zallocs$.

In each round, the algorithm computes a matching in the EFX feasibility graph. This matching can naturally be thought of as a partial allocation of the current bundles to the agents. If the matching is perfect, the corresponding (complete) allocation is returned and the algorithm terminates. Otherwise, an item is removed from some bundle and the algorithm proceeds to the next round. The matching computed in each round has particular properties that guarantee that the algorithm makes progress during its execution. In particular, all touched bundles are matched (condition (a) in Line 6). Under this condition, the matching contains the maximum number of edges of the form $(i,\zalloci)$ (condition (b) in Line 7). And, under these two conditions, the size of the matching is maximized (condition (c) in Line 8).

The algorithm finds the {\em robust demand} bundle of an arbitrary unmatched agent $i$ (Line 10) if the current matching is not perfect. The robust demand for agent $i$ is defined as any bundle
\[
\robustd{i,\zallocs}\eqdef\zalloci[j]\in\argmax_{\zalloc\in\zallocs}\{\max_{c\in \zalloc}\vali(\zalloc\setminus\{c\})\},
\]
breaking ties arbitrarily.
Then, the algorithm updates the set of bundles $\zallocs$ by removing from the robust-demand bundle $\zalloci[j]$ the least valued item of agent $i$ (Line 11). The definition of the robust demand guarantees that the bundle $\zalloci[j]$ will be EFX feasible for agent $i$ in the next round. In addition, as we will see, the edge $(i,\zalloci[j])$ will belong to the EFX feasibility graph in the next iteration.

%

\IncMargin{0.6em}
\begin{algorithm}[ht]
\KwIn{allocation $\allocs=(\alloci[1],\ldots,\alloci[n])$ of $\items$, such that $\nashw(\allocs)=\optnw(\items)$.}
\KwOut{allocation $\yallocs=(\yalloci[1],\ldots,\yalloci[n])$ of a set $S\subseteq\items$.}
Let $\zallocs=(\zalloci[1],\ldots,\zalloci[n])\gets(\alloci[1],\ldots,\alloci[n])$ be an ordered partition of $\items$.\\
\Repeat{$|M|=n$}
{
$G\gets\EFXgraph{[n],\zallocs}$;\\
Let $T$ be the set of touched bundles in $\zallocs$;\\
Let $M$ be a matching in $G$ such that\\
\quad\textbf{(a)}~~ All bundles in $T$ are matched in $M$,\\
\quad\textbf{(b)}~~ Under (a), $|M\cap\left\{(i,\zalloci): i\in[n]\right\}|$ is maximized, and\\
\quad\textbf{(c)}~~ Under (a) and (b), $|M|$ is maximized;\\
\If{$\exists i\in[n]$ not matched in $M$}
	{
		$\zalloci[j]\gets\robustd{i,\zallocs}$;\\
		$\zallocs \gets (\zalloci[1], \ldots, \zalloci[j-1],\zalloci[j]\setminus \{c\}, \zalloci[j+1], \ldots, \zalloci[n])$, where $c\in\{\argmin_{g\in\zalloci[j]}\vali(g)\}$;
	}
}
\KwRet{$\yallocs=\left(M(1),\ldots,M(n)\right)$} \Comment*{It must be the case $\{M(i)=\zalloci\}_{i\in [n]}$.}
\caption{Computes an EFX allocation of high Nash welfare.}
\label{alg:1}
\end{algorithm}
\DecMargin{0.6em}


Algorithm~\ref{alg:1} will eventually terminate and output some allocation $\yallocs$, since in every iteration of the repeat loop (except the last time $\bar{t}$) it removes an item from $\items$. Next, we need to argue that in every iteration $t\in\{0,\ldots,\bar{t}\}$ of the repeat loop the algorithm is correct, i.e., that in Line 6 of Algorithm~\ref{alg:1}, all touched bundles from $T(t)$ can be simultaneously matched in graph $G(t)\eqdef\EFXgraph{[n],\zallocs(t)}$\footnote{In the $t^{th}$ iteration, $\zallocs(t)$ represents $\zallocs$ before the execution of the if statement in Lines 9-12 of the algorithm.}. This is proved in Lemma~\ref{lem-exist-matching}; the proof uses several useful observations about the structure of the EFX feasibility graph $G(t)$.  
%
\begin{lemma}\label{lem-exist-matching}
All bundles in $T(t)$ can be simultaneously matched to agents in $G(t)$ at any iteration $t$.
\end{lemma}
\begin{proof} The proof proceeds by induction on $t$. For $t=0$, the statement is true, as $T(0)=\emptyset$ and we can choose $M(0)=\emptyset$. We assume that at time $t<\bar{t}$ the matching $M(t)$ covers all bundles from $T(t)$. Let us assume that agent $i^*$ is chosen in the execution of the if statement (Line $9$ of Algorithm~\ref{alg:1}) at time $t$ and the item $c^*\in\{\argmin_{g\in\zalloci[j^*](t)}\vali[i^*](g)\}$ is removed from the bundle $\zalloci[j^*](t)=\robustd{i^*,\zallocs(t)}$ in Lines 10-11. Now, in order to show that all bundles in $T(t+1)$ can be  matched to agents in $G(t+1)$, we will need a few useful observations about the edges of $G(t+1)$.

\begin{claim}
\label{cl:old_edges}
Graph $G(t+1)$ contains every edge $e=(i,\zalloci[j])$ with $\zalloci[j]\neq\zalloci[j^*]$ of the graph $G(t)$.
\end{claim}
\begin{proof} Indeed, since no item is removed from bundle $\zalloci[j]$, we have  $\vali(\zalloci[j](t+1))=\vali(\zalloci[j](t))\ge\vali(\zalloci(t))\ge\vali(\zalloci(t+1))$. Similarly, $\zalloci[j](t+1)$ remains EFX feasible for agent $i$. Thus $(i,\zalloci[j])\in E(G(t+1))$.
\end{proof}

\begin{claim}\label{claim-no-edge}
There is no edge between $i^*$ and $\zalloci[i^*]$ in $G(t)$.
\end{claim} 
\begin{proof} We prove this claim by contradiction. Suppose the edge $(i^*,\zalloci[i^*])$ belongs to $E(G(t))$, but $i^*$ was not matched in $M(t)$. Then we may increase the number of matched pairs $(i,\zalloci)$ in $M(t)$ by adding or possibly substituting another edge to $\zalloci[i^*]$ in $M(t)$. As the set of matched bundles in $\zallocs(t)$ can only increase after such an operation, we get a contradiction to condition (b) (Line 7 of Algorithm~\ref{alg:1}).
\end{proof} 

\begin{claim}\label{claim-edge}
There is an edge between $i^*$ and $\zalloci[j^*]$ in both graphs $G(t)$ and $G(t+1)$.
\end{claim}
\begin{proof}
Claim~\ref{claim-no-edge} says that $\zalloci[i^*](t)$ is not EFX feasible for agent $i^*$, i.e., there exists agent $i\in [n]$ such that $\vali[i^*](\zalloci[i^*](t))< \vali[i^*](\zalloci(t)\setminus \{g\})$ for some item $g\in\zalloci(t)$. As bundle $\zalloci[j^*](t)$ is the robust demand of $i^*$ at time $t$, we have 
\[
\vali[i^*](\zalloci[j^*](t+1))=\vali[i^*](\zalloci[j^*](t)\setminus\{c^*\})\ge\max_{c\in\zalloci(t)}\vali[i^*](\zalloci(t)\setminus\{c\})>\vali[i^*](\zalloci[i^*](t))\ge
\vali[i^*](\zalloci[i^*](t+1)).
\]
Similarly, $\zalloci[j^*](t)$ and $\zalloci[j^*](t+1)$ are EFX feasible for agent $i^*$, as $\vali[i^*](\zalloci[j^*](t)\setminus\{c^*\})$ is at least as high as $\vali[i^*](\zalloci[j](t)\setminus\{g\})$ for all $j\in[n]$ and any item $g\in\zalloci[j](t)$. Therefore, $(i^*,\zalloci[i^*])\in E(G(t+1))$ and $(i^*,\zalloci[i^*])\in E(G(t))$.
\end{proof}
We can now show that the algorithm will match $\zalloci[j^*]$ to some agent. 
\begin{claim}
\label{cl:matched_bundle}
Bundle $\zalloci[j^*]$ is matched in $M(t)$ to an agent $k$.
\end{claim}
\begin{proof}
If the bundle $\zalloci[j^*]$ is unmatched in $M(t)$, then there is a bigger matching $M(t)\cup (i^*,\zalloci[j^*])$ in $G(t)$: by Claim~\ref{claim-edge}, the edge $(i^*,\zalloci[j^*])$ belongs to $E(G(t))$ and, as the agent chosen when Line 9 is executed during iteration $t$, $i^*$ is unmatched in $M(t)$. This contradicts property (c) in Line 8 of Algorithm~\ref{alg:1}. 
\end{proof}

Finally, we note that $T(t+1)=T(t)\cup\zalloci[j^*]$, as bundle $\zalloci[j]$ is the only new bundle that can become touched in iteration $t$. Now consider the matching 
\be
\label{eq:useful_matching}
M'\eqdef M(t)\setminus\{(k,\zalloci[j^*])\}\cup\{(i^*,\zalloci[j^*])\}
\ee 
in $G(t+1)$ that covers all bundles in $T(t+1)$. Indeed, according to Claim~\ref{cl:old_edges}, $G(t+1)$ contains all edges of $M(t)\setminus \{(k,\zalloci[j^*])\}$ and, according to Claim~\ref{claim-edge}, the edge $(i^*,\zalloci[j^*])$ belongs to $G(t+1)$ as well.
Hence, $M'$ is a matching in $G(t+1)$. Since (by our induction hypothesis) all bundles of $T(t)\setminus \{\zalloci[j^*]\}$ are matched in $M(t)\setminus \{(k,\zalloci[j^*])\}$, $M'$ covers all bundles in $T(t+1)=T(t)\cup\zalloci[j^*]$. The proof of the induction step is complete.
\end{proof}

As the algorithm terminates after finding a complete matching in the EFX feasibility graph, the returned solution $\yallocs$ must be an EFX allocation. We note that Algorithm~\ref{alg:1} runs in polynomial time: it executes at most $m$ iterations and all steps in each iteration can be completed in polynomial time. Indeed, the only non trivial part is the computation of matching $M$ (Line 5) under the conditions (a)-(c) (Lines 6-8). To this end, we can assign weights to the edges of the graph $G$ that express our preferences (a)-(c) and compute a maximum weighted matching in the weighed version of $G$.\footnote{We remark that this is apparently not the fastest way to compute the desired matching.} For example, we can give a weight of $n^4$ to the edges of $G$ incident to bundle $\zalloci\in T$, a weight of $n^2$ to edges $(i,\zalloci)\in E(G)$ for $i\in [n]$, and weight of $1$ to the remaining edges of $G$.

In the following we give efficiency guarantees for the returned EFX allocation $\yallocs$. Very informally, Algorithm~\ref{alg:1} does not remove too many items. First, we observe that at the end of the algorithm there must be at least one untouched bundle, i.e., a bundle from which no items have been removed.   

\begin{claim}
\label{cl:untouched_bundle}
There is an untouched bundle $\zalloci=\alloci$ upon termination of Algorithm~\ref{alg:1}. 
\end{claim}
\begin{proof}
Assume to the contrary that the algorithm has removed items from every bundle in $\zallocs$ after iteration $t$, i.e., the set of touched bundles becomes the set $T(t+1)=\{\zalloci\}_{i\in[n]}$ of size $n$ after iteration $t$. Then, by Lemma~\ref{lem-exist-matching}, the algorithm will match all bundles in $M(t+1)$ and terminate at round $t+1$. Let agent $i^*$ and bundle $\zalloci[j^*]$ be the ones chosen, respectively, when Algorithm~\ref{alg:1} executes Lines 9 and 10 at iteration $t$. At this time, we have $\zalloci[j]\in T(t)$ for all $j\neq j^*$. Thus, all touched bundles $\zalloci[j]$ with $j\neq j^*$ must be matched in $M(t)$ (due to the condition (a) in Line 6). According to Claim~\ref{cl:matched_bundle}, bundle $\zalloci[j^*]$ must also be matched in $M(t)$. Therefore, $|M(t)|=n$ and the algorithm should have terminated after iteration $t$ leaving bundle $\zalloci[j^*]$ untouched.
%
\end{proof}

We are now ready to present guarantees for the Nash welfare of allocation $\yallocs$. 

\begin{lemma}
\label{lem:half_nsw_max}
$\vali(\alloci)\le 2\cdot\vali(\yalloci)$ for any agent $i\in[n]$ and there exists an agent $i^o\in [n]$ such that $\vali[i^o](\alloci[i^o])\le\vali[i^o](\yalloci[i^o])$.
\end{lemma}
\begin{proof} We note first that when the algorithm terminates at time $\bar{t}$, then $\vali(\zalloci)\le\vali(\yalloci)$ for all $i\in[n]$. Indeed, this is the case if agent $i$ is matched to bundle $\zalloci$. If $i$ is matched to another bundle $\yalloci=\zalloci[j]$ in $G(\bar{t})$, then $\vali(\zalloci[j])>\vali(\zalloci)$ by the definition of the EFX feasibility graph $G(\bar{t})$. Thus,  Claim~\ref{cl:untouched_bundle} proves the second part of the lemma since its states that there exists some agent $i^o$ such that $\vali[i^o](\alloci[i^o])=\vali[i^o](\zalloci[i^o])\le\vali[i^o](\yalloci[i^o])$. 
	
To complete the proof of the lemma, it is sufficient to show that $2\cdot\vali(\zalloci)\ge\vali(\alloci)$ for any agent $i\in [n]$. We do so by contradiction. If this is not the case, let $t\in\{0,\ldots,\bar{t}\}$ be the first time when $2\cdot\vali(\zalloci(t))$ becomes strictly smaller than $\vali(\alloci)$ for an agent $i$ after an item was removed. Let $i^*$ be the agent who was chosen in Line $9$ and $\zalloci[j^*]$ be the bundle chosen in Line $10$ of the algorithm at iteration $t$. Then $j^*$ is the agent for whom $2\cdot\vali[j^*](\zalloci[j^*](t))$ became smaller than $\vali[j^*](\alloci[j^*])$. For convenience of notation, we denote the matching $\{(i,\zalloci)\}_{i\in[n]}$ by $M^o$. We consider the matching $M'\bydef M(t)\setminus\{(k,\zalloci[j^*])\}\cup\{(i^*,\zalloci[j^*])\}$, that was also used in the proof of Lemma~\ref{lem-exist-matching}. Recall that, by Claim~\ref{cl:matched_bundle}, bundle $\zalloci[j]$ is matched to agent $k$ in $M(t)$. As agent $k$ is not matched in $M'$, we have $|M'|<n=|M^o|$. Therefore, one can represent the union of the matchings $M'$ and $M^o$ as a collection of augmenting paths and cycles, including degenerate cycles that consist of a single edge that belongs to both $M'$ and $M^o$. 

First, assume that $M'\cup M^o$ has a non-degenerate cycle $C$ with $2k$ vertices $\{j_1,\ldots,j_k\}$ and $\{\zalloci[j_1],\ldots,\zalloci[j_k]\}$.
\[
C\eqdef\quad\{(j_1,\zalloci[j_2]),\ldots,(j_{k-1},\zalloci[j_k]),(j_k,\zalloci[j_1])\}\in M',\quad\quad\{(j_1,\zalloci[j_1]),\ldots,(j_k,\zalloci[j_k])\}\in M^o.
\]
For convenience of notation, let us denote $j_{k+1}=j_1$. Since each $(j_i,\zalloci[j_{i+1}])\in M'$ is an edge in the EFX feasibility graph $G(t)$, we have $\vali[j_i](\zalloci[j_{i+1}])>\vali[j_i](\zalloci[j_i])$ for every $i\in[k]$. This implies that $\vali[j_i](\alloci[j_i])+\vali[j_i](\zalloci[j_{i+1}])-\vali[j_i](\zalloci[j_i])>\vali[j_i](\alloci[j_i])$ for every $i\in[k]$. Thus, we get a Pareto improvement in the allocation $\allocs$ for each agent $j_i\in C$ by replacing bundle $\alloci[j_i]$ with bundle $\alloci[j_i]\setminus\zalloci[j_i]\cup\zalloci[j_{i+1}]$ for every $i\in[k]$. Hence, we can improve the Nash welfare of the allocation $\allocs$, a contradiction to its optimality.

Second, we observe that bundle $\zalloci[j^*]$ is matched in $M'$ to $i^*$ which, by Claim~\ref{claim-no-edge}, must be different from $j^*$. Since $M'\cup M^o$ does not have non-degenerate cycles and bundle $\zalloci[j^*]$ is not matched to $j^*$ in $M'$, bundle $\zalloci[j^*]$ must belong to an augmenting path $P$ that originates from an unmatched bundle $\zalloci[j_1]$ in $M'$; notice that $\zalloci[j_1]$ is unmatched in $M(t)$ as well, since matching $M(t)$ has the same set of bundles as $M'$. The augmenting path $P$ consists of $k$ agent vertices $j_1,\ldots,j_k$ and $k+1$ bundle vertices $\zalloci[j_1],\ldots,\zalloci[j_{k+1}]$, where $j_{k+1}=j^*$, and $2k$ edges:
\be
\label{eq:augmenting_path}
P\eqdef\quad\{(j_1,\zalloci[j_2]),\ldots,(j_k,\zalloci[j_{k+1}])\}\in M',
\quad\quad\{(j_1,\zalloci[j_1]),\ldots,(j_k,\zalloci[j_k])\}\in M^o.
\ee

We consider the following transformation $\hallocs$ of the initial allocation $\allocs$ (see Figure~\ref{fig:augmenting}):
\[
\hallocs\eqdef
\begin{cases}
\halloci[j_1]=\alloci[j_1]\cup\zalloci[j_{2}], & \\
\halloci[j_i]=\alloci[j_i]\setminus\zalloci[j_i]\cup\zalloci[j_{i+1}], &i\in\{2,...,k\}\quad\\
\halloci[j_{k+1}]=\alloci[j_{k+1}]\setminus\zalloci[j_{k+1}] &\\
\halloci[j]=\alloci[j], & j\notin\{j_1,...,j_{k+1}\}
\end{cases}
\]

\begin{figure}[h]
	\centering
	\includegraphics[width=0.9\textwidth]{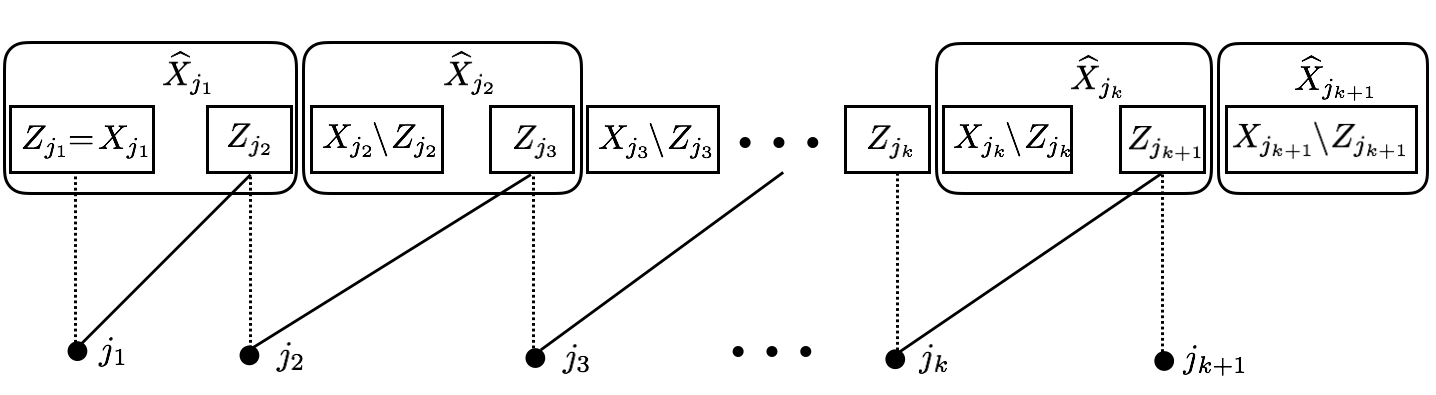}
	\caption{The main argument in the proof of Lemma~\ref{lem:half_nsw_max}: reallocating the items of allocation $\allocs$ to obtain another allocation $\hallocs$ with improved Nash welfare.}	
	\label{fig:augmenting}	
\end{figure}

Let us denote the set of agents $\{j_1,\ldots,j_{k+1}\}$ as $J$. Similar to the previous case where we consdiered a non-degenerate cycle, we have $\vali[j_i](\zalloci[j_{i+1}])>\vali[j_i](\zalloci[j_i])$ for every $i\in[k]$. By our assumption for agent $j^*$ (with $j^*=j_{k+1}$), it holds that  $\vali[j_{k+1}](\zalloci[j_{k+1}])<\frac{1}{2}\cdot \vali[j_{k+1}](\alloci[j_{k+1}])$. Finally, since $\zalloci[j_1]$ is unmatched in $M(t)$ it must also be an untouched bundle, i.e., $\zalloci[j_1]=\alloci[j_1]$. Combining these properties, we get:
\begin{align*}
\nashw(\hallocs)^n &=
\vali[j_1](\alloci[j_1]\cup\zalloci[j_{2}])\cdot
\vali[j_{k+1}](\alloci[j_{k+1}]\setminus\zalloci[j_{k+1}])\cdot
\prod_{i=2}^{k}\vali[j_i](\alloci[j_i]\setminus\zalloci[j_i]\cup\zalloci[j_{i+1}])\cdot
\prod_{j\notin J}\vali[j](\alloci[j])\\
&= \left[\vali[j_1](\alloci[j_1])+\vali[j_1](\zalloci[j_{2}])\right] \cdot
\left[\vali[j_{k+1}](\alloci[j_{k+1}])-\vali[j_{k+1}](\zalloci[j_{k+1}])\right]\\
&\quad\quad \cdot
\prod_{i=2}^{k}{\left[\vali[j_i](\alloci[j_i])-\vali[j_i](\zalloci[j_i])+\vali[j_i](\zalloci[j_{i+1}])\right]} \cdot
\prod_{j\notin J}\vali[j](\alloci[j])\\
&> \left[\vali[j_1](\alloci[j_1])+\vali[j_1](\zalloci[j_1])\right]\cdot
\frac{1}{2}\vali[j_{k+1}](\alloci[j_{k+1}])\cdot
\prod_{i=2}^{k}\vali[j_i](\alloci[j_i])\cdot
\prod_{j\notin J}\vali[j](\alloci[j])=
\nashw(\allocs)^n
\end{align*}
We have reached a contradiction, since $\nashw(\allocs)=\optnw(\items)$.
\end{proof}

Using Lemma~\ref{lem-exist-matching} and~\ref{lem:half_nsw_max}, we obtain the following statement.

\begin{theorem}\label{thm:main}
Given a Nash welfare maximizing allocation $\allocs$, Algorithm~\ref{alg:1} computes in polynomial time a $2^{1-\frac{1}{n}}$-efficient EFX and Pareto-optimal allocation $\yallocs$ of $\cup_{i=1}^{n}\yalloci,$ such that $\yalloci\subseteq\alloci$ for all $i\in\agents$.
\end{theorem}

\begin{proof} 
The correctness of the algorithm was shown in Lemma~\ref{lem-exist-matching}. According to the definition of the EFX feasibility graph $G$, the final allocation, which is a complete matching in $G$, must be an EFX allocation. Finally, according to Lemma~\ref{lem:half_nsw_max}, the Nash welfare of the final allocation $\yallocs$ is at least a $2^{-(n-1)/n}$-fraction of $\optnw(\items)=\nashw(\allocs)$. 

Moreover, if the returned complete matching $M$ is different from the matching $M^o\bydef\{(i,\zalloci)\}_{i\in [n]}$, then there must be a cycle $C:\{(j_1,\zalloci[j_2]),...,(j_{k-1},\zalloci[j_k]),(j_k,\zalloci[j_1])\}\in M,\quad\{(j_1,\zalloci[j_1]),\ldots,(j_k,\zalloci[j_k])\}\in M^o$ in $M\cup M^o$. Then one can get a strict improvement to the initial Nash welfare maximizing allocation $\allocs$, by setting $\halloci[j_i]=\alloci[j_i]\setminus\zalloci[j_i]\cup\zalloci[j_{i+1}],$ where $j_{k+1}\eqdef j_1$. Indeed, $\vali[j_i](\zalloci[j_{i+1}])>\vali[j_i](\zalloci[j_i])$ and $\vali[j_i](\halloci[j_i])>\vali[j_i](\alloci[j_i])$ for all $i\in[k]$, as $(j_i,\zalloci[j_{i+1}])$ is an edge in the EFX feasibility graph $G$. Hence, $M(i)=\zalloci$ for all $i\in[n]$ and $\yalloci\subseteq\alloci$. 

Finally, allocation $\yallocs$ is Pareto optimal, as otherwise one could first Pareto improve the allocation $\yallocs$ and then give back to each agent $i$ their original removed items $\alloci\setminus\yalloci$; the resulting allocation would be a Pareto-improvement of the initial allocation $\allocs$, contradicting its optimality.
\end{proof}

\section{Suboptimal input allocation}
\label{sec:poly-time-algorithm}
One might hope that the algorithm we presented in Section~\ref{sec:algorithm} could also work with an input allocation of suboptimal Nash welfare. Indeed, the algorithm would result in an EFX allocation in this case too. Unfortunately, the proof of the efficiency guarantee in Lemma~\ref{lem:half_nsw_max}  crucially relies on the optimality of the input allocation. In this section we present a modified Algorithm~\ref{alg:2} with a good efficiency guarantee,
provided that the input allocation is efficient as well, albeit not necessarily optimal.
In particular, starting with an initial allocation $\allocs$, the modified algorithm computes either an EFX allocation $\yallocs$ of some of the items in $\items$ that recovers a large fraction of the Nash welfare of the initial allocation or indicates that the Nash welfare of the input allocation can be improved significantly and provides such an improved allocation $\hallocs$ as output. Repeating the algorithm with the allocation $\hallocs$ as input until it produces an EFX allocation will yield, in polynomial time, an EFX allocation with at least half the Nash welfare of the initial input allocation.

\IncMargin{0.6em}
\begin{algorithm}[h]
	\KwIn{allocation $\allocs=(\alloci[1],\ldots,\alloci[n])$ of $\items$ and $\delta\in\R_+$.}
	\KwOut{allocation $\yallocs=(\yalloci[1],\ldots,\yalloci[n])$ of $S\subseteq\items$ or allocation $\hallocs=(\halloci[1],\ldots,\halloci[n])$ of $\items$.}
	Let $\zallocs=(\zalloci[1],\ldots,\zalloci[n])\gets(\alloci[1],\ldots,\alloci[n])$ be an ordered partition of $\items$;\\
	Let $M_0=\{(i,\zalloci):i\in[n]\}$;\\
	Let $\delta_1 = \frac{2\delta}{1-\delta}$\Comment*{$\frac{2+2\delta_1}{2+\delta_1}=1+\delta$} 
	\Repeat{$|M|=n$}
	{
		$G\gets\EFXgraph{[n],\zallocs}$;\\
		Let $T$ be the set of touched bundles in $\zallocs$;\\
		Compute a matching $M$ in $G$ such that\\
		\quad\textbf{(a)}~~ All bundles in $T$ are matched in $M$,\\
		\quad\textbf{(b)}~~ Under (a), $|M\cap M_0|$ is maximized, and\\
		\quad\textbf{(c)}~~ Under (a) and (b), $|M|$ is maximized;\\
		Let $\zalloci[j_1]$ be an unmatched bundle in $M$;\\
		\Repeat{$\itrem$}
		{
			$\itrem \gets \false$;\\
			Let $P=\{(j_1,\zalloci[j_2]), \ldots, (j_{k-1},\zalloci[j_k])\}  \gets \AugmentingPath{\zalloci[j_1],M\cup M_0}$;\\
			$\zalloci[j^*]\gets\robustd{j_k,\zallocs}$;\\
			\eIf{$\exists (j^o,\zalloci[j^*])\in P$}
			{ 
				$M \gets M\setminus \{(j^o,\zalloci[j^*])\} \cup \{(j_k,\zalloci[j^*])\}$;\\
			}{
				$\zallocs \gets (\zalloci[1], \ldots, \zalloci[j^*-1],\zalloci[j^*]\setminus \{c\}, \zalloci[j^*+1], \ldots, \zalloci[n])$, where $c\in\{\argmin_{g\in\zalloci[j^*]}\vali[j_k](g)\}$;\\ 
				\If{$(2+\delta_1)\vali[j^*](\zalloci[j^*])<\vali[j^*](\alloci[j^*])$}
				{
					$\halloci[j_1]\gets \alloci[j_1]\cup\zalloci[j_2]$;\\
					$\halloci[j_i]\gets \alloci[j_i]\setminus\zalloci[j_i]\cup\zalloci[j_{i+1}]$ for $i=2, \ldots k-1$;\\
					$\halloci[j_k]\gets \alloci[j_k]\setminus\zalloci[j_k]\cup\zalloci[j^*]$;\\
					$\halloci[j^*]\gets \alloci[j^*]\setminus\zalloci[j^*]$;\\		
					$\halloci[i]\gets \alloci[i]$ for $i\in [n]\setminus \{j_1, j_2, ..., j_k, j^*\}$.\\
					\KwRet{$\hallocs$};
				}
				$\itrem \gets \true$;
			}
		}	
	}
	\KwRet{$\yallocs=\left(M(1),\ldots,M(n)\right)$};
	\caption{Computes an EFX allocation of high Nash welfare or indicates that a significant improvement to the Nash welfare of the input allocation is possible.}
	\label{alg:2}
\end{algorithm}
\DecMargin{0.6em}


The general structure of the algorithm is the same as before. It proceeds in rounds (defined by the outer repeat-until loop; Lines 5-30). In each round, it tries to match as many bundles from the initial allocation to agents as possible (by essentially repeating the matching computation on the EFX feasibility graph; Lines 5-10) and repeatedly removes items from the bundles as long as this matching does not correspond to an EFX allocation. The main difference of the modified algorithm is in the selection of the agent and an item in their robust demand bundle to be removed at each round, which is implemented in lines 11-30 and includes the inner repeat-until loop. This step is more complicated now and may modify the computed matching as well. We exploit paths on the union of matching $M$ with the identity matching $M^o\bydef \{(i,\zalloci):i\in[n]\}$. In particular, given an unmatched bundle $\zalloci[j_1]$ (defined in line 11), the algorithm first computes the augmenting path of $M\cup M^o$ that originates from vertex $\zalloci[j_1]$ of the EFX feasibility graph. This is a path that alternates between edges of matching $M^o$ and $M$, and terminates with an edge of $M^o$ and the vertex corresponding to the unmatched agent $j_k$. The algorithm temporarily selects agent $j_k$ (Line 14, where the set of edges of $M$ that belong to the augmenting path is returned by the call to the subroutine $\augmentingpath$) and computes agent's $j_k$ robust demand (Line 15), say bundle $\zalloci[j^*]$, but it does not remove any item from it yet. 

If bundle $\zalloci[j^*]$ belongs to the augmenting path and appears in the edge $(j^o,\zalloci[j^*])$ of $M$, the algorithm modifies $M$ by removing edge $(j^o,\zalloci[j^*])$ and adding edge $(j_k,\zalloci[j^*])$ to it (Line 17). We will see shortly that this is a valid modification that does not violate the properties (a), (b), and (c) of matching $M$. The algorithm repeats the augmenting path process until the computed robust demand bundle $\zalloci[j^*]$ does not belong to the augmenting path. In this case, the execution flow enters the else statement and the least valued item for agent $j_k$ is removed from bundle $\zalloci[j^*]$ (Line 19). Before completing the current round, the algorithm checks whether the value of agent $j^*$ for bundle $\zalloci[j^*]$ has dropped significantly below half (more than by a factor of $2+\delta_1$) of her value for the initial bundle $\alloci[j^*]$ (Line 20). If this is the case, the algorithm computes a new allocation $\hallocs$ (Lines 21-25) which, as we will prove, has Nash welfare at least $(1+\delta)^{1/n}\nashw(\allocs)$. It then terminates, returning allocation $\hallocs$ as output. Otherwise, it indicates the end of the item removal process (Line 28), which will allow execution flow to leave the inner repeat-until loop.

For the analysis, we first claim that Lemma~\ref{lem-exist-matching} carries over to the modified algorithm. All we need to show is that matching $M$ keeps all its properties (a),(b), and (c) after every modification (in Line 17). As the modification affects neither the set of matched bundles, nor the size of $M\cap M^o$, it suffices to show that the new edge $(j_k,\zalloci[j^*])$ does exist in the EFX feasibility graph. The proof is similar to the proof of Claim~\ref{claim-edge}. Notice that the edge $(j_k,\zalloci[j_k])$ is not EFX feasible, since otherwise it could replace edge $(j_{k-1},\zalloci[j_k])$ in $M$ to increase $|M\cap M^o|$, still matching all touched bundles. Hence, there exists $j$ such that  
$$\vali[j_k](\zalloci[j_k])<\max_{g\in \zalloci[j]}{\vali[j_k](\zalloci[j]\setminus\{g\})} \leq \max_{g\in \zalloci[j^*]}{\vali[j_k](\zalloci[j^*]\setminus\{g\})}\leq \vali[j_k](\zalloci[j^*]),$$
where the second inequality follows since $\zalloci[j^*]$ is the robust demand of agent $j_k$. Hence, the edge $(j_k,\zalloci[j^*])$ indeed belongs to the EFX feasibility graph.
	
	
Now, if the algorithm terminates by returning allocation $\yallocs$ in Line 32, this will be an EFX allocation satisfying $(2+\delta_1)\cdot \vali(\yalloci)\geq \vali(\alloci)$ for every agent $i\in[n]$. Furthermore, the argument in the proof of Claim~\ref{cl:untouched_bundle} carries over, and one of the bundles, say $\alloci[i^o]$, will stay untouched until the end of the execution so that $\vali(\yalloci[i^o])\geq \vali[i^o](\alloci[i^o])$. Consequently, the Nash welfare of allocation $\yallocs$ is
$$\left(2+\delta_1\right)^{1-1/n}\cdot \nashw(\yallocs)=\left(\vali[i^o](\yalloci[i^o])\cdot \prod_{i\in[n]\setminus\{i^o\}}{(2+\delta_1)\vali(\yalloci)}\right)^{1/n} \geq \left(\prod_{i\in[n]}{\vali(\alloci)}\right)^{1/n}= \nashw(\allocs).$$
This is summarized in the following statement.
\begin{lemma}\label{lem:y}
	If the algorithm terminates and outputs the EFX allocation $\yallocs$, then $(2+\delta_1)^{1-1/n} \cdot \nashw(\yallocs)\geq \nashw(\allocs)$.
\end{lemma}	

If, in contrast, the execution of the algorithm enters Lines 21-25 and the algorithm terminates returning the allocation $\hallocs$, we can show the following.

\begin{lemma}\label{lem:hat-x}
	If the algorithm terminates and outputs allocation $\hallocs$, then $\nashw(\hallocs)\geq (1+\delta)^{1/n}\nashw(\allocs)$.
\end{lemma}	
\begin{proof}
When the execution enters lines 21-26, it holds that $(2+\delta_1)\vali[j^*](\zalloci[j^*])<\vali[j^*](\alloci[j^*])$, which implies that $(2+\delta_1)\vali[j^*](\halloci[j^*])>(1+\delta_1)\vali[j^*](\alloci[j^*]).$ As $\vali[j_i](\zalloci[j_{i}])\leq \vali[j_i](\zalloci[j_{i+1}])$ for $i=2, ..., k-1$ and $\vali[j_k](\zalloci[j_{k}])\leq \vali[j_k](\zalloci[j^*])$, we also have $\vali(\halloci)\geq \vali(\alloci)$ for $i\in [n]\setminus\{j_1,j^*\}$. Furthermore, $\vali[j_1](\halloci[j_1])\geq 2\vali[j_1](\alloci[j_1]).$
Putting the above inequalities together, we have
$$\nashw(\hallocs) =\left(\prod_{i\in[n]}{\vali(\halloci)}\right)^{1/n} \geq \left(\frac{2(1+\delta_1)}{2+\delta_1}\right)^{1/n}\left(\prod_{i\in[n]}{\vali(\alloci)}\right)^{1/n}= (1+\delta)^{1/n}\nashw(\allocs).$$
The last equality is due to the relation of parameters $\delta$ and $\delta_1$ (Line 3).
\end{proof}

We keep running Algorithm~\ref{alg:2}, starting with the initial input allocation $\allocs$, and, every time it outputs an allocation $\hallocs$ with significantly higher Nash welfare than $\allocs$, we let $\allocs\gets\hallocs$ and invoke it again. We stop when Algorithm~\ref{alg:2} outputs an EFX allocation $\yallocs$ for the first time. By Lemma~\ref{lem:y}, this is a $(2+\delta_1)^{1-1/n}$ approximation to the Nash welfare of the current input allocation $\allocs$. As the Nash welfare of the input allocation only improves with time, we get $(2+\delta_1)^{1-1/n}\cdot\rho\cdot\nashw(\yallocs)\ge\optnw(\items)$. Notice that no more than $\Ocomplex{\frac{n\rho}{\delta}}$ executions of the algorithm will be required before Algorithm~\ref{alg:2} outputs EFX allocation $\yallocs$. Indeed, by Lemma~\ref{lem:hat-x}, we know that, after every $n$ updates of $\allocs\gets\hallocs$, the Nash welfare of $\allocs$ increases by a factor of at least $(1+\delta)$, yielding an additive improvement to the Nash welfare of at least $\delta\cdot\nashw(\allocs)\ge\frac{\delta}{\rho}\optnw(\items)$. Thus, if the Algorithm~\ref{alg:2} does not output $\yallocs$, $\nashw(\allocs)$ would become larger than $\optnw(\items)$ after $\Ocomplex{\frac{n\rho}{\delta}}$ updates $\allocs\gets\hallocs$. Setting $\delta=\frac{1}{2n+1}$, and, consequently, $\delta_1=\frac{1}{n}$ we obtain a $2\rho$-approximation\footnote{We get $\left(2+\delta_1\right)^{1-1/n}\cdot\rho$-efficiency for the allocation $\yallocs$, where $(2+\frac{1}{n})^{1-1/n}<2$.} after at most $\Ocomplex{n^2}$ executions of Algorithm~\ref{alg:2}. The following statement summarizes the discussion of this section.


\begin{theorem}
On input a $\rho$-efficient allocation, 
Algorithm~\ref{alg:2} returns a $2\rho$-efficient EFX allocation
 after at most $O(n^2)$ repeated executions.
\end{theorem}

\section{Large markets}
\label{sec:large_market}
In this section, we give an improved guarantee for the Nash social welfare of the EFX allocation $\yallocs$ produced by Algorithm~\ref{alg:1}, if the initial Nash social welfare maximizing allocation $\allocs$ satisfies the large market condition with a parameter $\eps$.

\begin{theorem}
\label{th:large_market}
If the input allocation $\allocs$ satisfies the large market condition with a parameter $\eps$, then Algorithm~\ref{alg:1} outputs a $\left(1+8\sqrt{\eps}\right)$-efficient EFX allocation $\yallocs$.
\end{theorem}
\begin{proof}
In order to show the stated guarantee, we prove a stronger version of Lemma~\ref{lem:half_nsw_max} where, instead of the guarantee $\vali(\alloci)\le 2\cdot\vali(\yalloci)$, we show that $\vali(\alloci)\le \left(1+8\cdot\sqrt{\eps}\right)\vali(\yalloci)$, for each agent $i\in[n]$. We follow the approach in the proof of Lemma~\ref{lem:half_nsw_max}, with the only difference that we stop at the iteration $t$ of Algorithm~\ref{alg:1} when the condition 
\be
\label{eq:condition_eps}
\forall i\in[n]\quad\quad\left(1+8\cdot\sqrt{\eps}\right)\vali(\zalloci(t))\ge\vali(\alloci).
\ee 
is violated for the first time (instead of the condition $2\cdot\vali(\zalloci(t))\ge\vali(\alloci)$ that we considered in the proof of Lemma~~\ref{lem:half_nsw_max}).

As in Lemma~\ref{lem:half_nsw_max}, if condition~\eqref{eq:condition_eps} is violated for the bundle $\zalloci[j^*]$, we can find an augmenting path $P$ defined in equation~\eqref{eq:augmenting_path} with a set of vertices $\{j_1,\ldots,j_k\}\subset\agents$ and $\{\zalloci[j_1],\ldots,\zalloci[j_{k+1}]\}\subseteq\{\zalloci\}_{i=1}^{n}$ from an unallocated bundle $\zalloci[j_1]$ to the violating bundle $\zalloci[j_{k+1}]=\zalloci[j^*]$. To simplify notation, we assume without loss of generality that $j_i=i$ for each $i\in[k+1]$. Hence, we have
\[
P\bydef\quad\{(1,\zalloci[2]),\ldots,(k,\zalloci[k+1])\}\in M',\quad\quad\{(1,\zalloci[1]),\ldots,(k,\zalloci[k])\}\in M^o.
\]
We recall some key facts about the bundles in the augmenting path $P$. Since $(i,\zalloci[i+1])$ is an edge in the EFX feasibility graph $G$, $\vali(\zalloci[i+1])>\vali(\zalloci)$ for all $i\in[k]$. Since condition \eqref{eq:condition_eps} is violated for the bundle $\zalloci[k+1]$, $\left(1+8\cdot\sqrt{\eps}\right)\vali[k+1](\zalloci[k+1])<\vali[k+1](\alloci[k+1])$. Also, since bundle $\zalloci[1]$ is unallocated, it must be untouched, i.e., $\zalloci[1]=\alloci[1]$. 

Next, we consider how to reallocate items between each consecutive pair of bundles $\zalloci$ and $\zalloci[i+1]$. Unlike in the proof of Lemma~\ref{lem:half_nsw_max}, the reallocation of the entire set $\zalloci[i+1]$ to each agent $i$ does not necessarily increase the Nash social welfare. To obtain an improvement, we have to be more flexible and explore different reallocation schemes from agent $i+1$ to agent $i$. Among many possible schemes, we employ the greedy strategy. Specifically, we rename the items in $\zalloci[i+1]$ so that $\zalloci[i+1]=\{g_1,\ldots,g_\ell\}$ with
\[
\frac{\vali(g_1)}{\vali[i+1](g_1)}\ge\frac{\vali(g_2)}{\vali[i+1](g_2)}\ge\cdots\ge\frac{\vali(g_\ell)}{\vali[i+1](g_\ell)}.
\]
In the greedy scheme, we consider to move the first $r$ items, i.e., the set $R(r)=R_{i+1}(r)\eqdef\{g_1,\ldots,g_r\}$, from agent $i+1$ to agent $i$. The following Claim establishes an important property of any such move: relative increase in agent $i$'s value is greater than the relative decrease in agent $i+1$'s value.

\begin{claim}
\label{cl:greedy_reallocation}
$\frac{\vali(R(r))}{\vali(\zalloci)}\ge\frac{\vali[i+1](R(r))}{\vali[i+1](\zalloci[i+1])}$ for any $r\in\{0,\ldots,\ell\}$. 
\end{claim}

\begin{proof}
As $\vali(\zalloci[i+1])>\vali(\zalloci)$, we have $\frac{\vali(R(r))}{\vali(\zalloci)}\ge\frac{\vali(R(r))}{\vali(\zalloci[i+1])}.$ It remains to show that $\frac{\vali(R(r))}{\vali(\zalloci[i+1])}\ge\frac{\vali[i+1](R(r))}{\vali[i+1](\zalloci[i+1])}$ or, equivalently, $\vali(R(r))\vali[i+1](\zalloci[i+1])\ge\vali(\zalloci[i+1])\vali[i+1](R(r))$. Let $S=\zalloci[i+1]\setminus R(r)$; then we need to show $\vali(R(r))(\vali[i+1](S)+\vali[i+1](R(r)))\ge(\vali(S)+\vali(R(r)))\vali[i+1](R(r))$ or, equivalently, $x\eqdef\vali(R(r))\vali[i+1](S)-\vali(S)\vali[i+1](R(r))\ge 0$. 
By the definition of $x$, $R(r)$, and the renaming of items in $Z_{i+1}$, we have that $x=\sum_{j=1}^{r}\sum_{s=r+1}^{\ell}\left(\vali(g_j)\vali[i+1](g_s)-\vali(g_s)\vali[i+1](g_j)\right)$.
 Since $\frac{\vali(g_j)}{\vali[i+1](g_j)}\ge\frac{\vali(g_s)}{\vali[i+1](g_s)}$ for any $j\ge r \ge s$, we have $\vali(g_j)\vali[i+1](g_s)-\vali(g_s)\vali[i+1](g_j)\ge 0$. Therefore, $x\ge 0$ 
\end{proof}

Claim~\ref{cl:greedy_reallocation} gives us a powerful tool to balance the valuations of agent $i$ and $i+1$. The major challenge, however, is to {\em simultaneously} adjust the valuations of the agents $i\in[k+1]$ in a way that improves the Nash social welfare of the initial allocation $\allocs$. We consider arbitrary simultaneous reallocation of items between all pairs of agents $i$ and $i+1$, where the first $r_{i+1}$ items are moved from agent $i+1$ to agent $i$ according to the greedy scheme for all $i\in[k]$. 
We can focus only on the changes of the allocations for the first $k+1$ agents, as the remaining bundles $\alloci$ remain unchanged for $i\notin[k+1]$. The new allocation $\hallocs$ for the first $k+1$ agents is as follows.
\[
\hallocs:
\begin{cases}
\halloci[1]  &\eqdef\alloci[1]\cup R_2(r_2) \\
\halloci     &\eqdef\alloci\setminus R_i(r_i)\cup R_{i+1}(r_{i+1}), \quad i\in\{2,...,k\}\\
\halloci[k+1]&\eqdef\alloci[k+1]\setminus R_{k+1}(r_{k+1}) 
\end{cases}
\]
We compare the Nash social welfare of the new allocation $\hallocs$ with that of the initial allocation $\allocs$.
\be
\label{eq:nash_welfare_lm}
\ln \left(\frac{\nashw(\hallocs)^n}{\nashw(\allocs)^n}\right)=\ln\left(\prod_{i=1}^{k+1}\frac{\vali(\halloci)}{\vali(\alloci)}\right)=
\sum_{i=1}^{k+1}\ln\left(\frac{\vali(\halloci)}{\vali(\zalloci)}\right)-\sum_{i=1}^{k+1}\ln\left(\frac{\vali(\alloci)}{\vali(\zalloci)}\right).
\ee
As every bundle $\zalloci$ is always a subset of the initial bundle $\alloci$, we have $\vali(\alloci)=\vali(\zalloci)+\vali(\alloci\setminus\zalloci)$ and $\frac{\vali(\alloci)}{\vali(\zalloci)}=1+\frac{\vali(\alloci\setminus\zalloci)}{\vali(\zalloci)}$. Let $a_i\eqdef\frac{\vali(\alloci\setminus\zalloci)}{\vali(\zalloci)}$, so $\frac{\vali(\alloci)}{\vali(\zalloci)}=1+a_i$ for each $i\in[k+1]$. To simplify notation, let $R_1(r_1),R_{k+2}(r_{k+2})\eqdef\emptyset$, so that $\halloci=\alloci\setminus R_i(r_i)\cup R_{i+1}(r_{i+1})$ for each $i\in[k+1]$ and $\frac{\vali(\halloci)}{\vali(\zalloci)}=1+\frac{\vali(\alloci\setminus\zalloci)}{\vali(\zalloci)}-\frac{\vali(R_i(r_i))}{\vali(\zalloci)}+\frac{\vali(R_{i+1}(r_{i+1}))}{\vali(\zalloci)}$. Let $x_i\eqdef\frac{\vali(R_i(r_i))}{\vali(\zalloci)}$ for $i\in[k+1]$; then we get
\[
\frac{\vali(\halloci)}{\vali(\zalloci)}=1+a_i-x_i+\frac{\vali(R_{i+1}(r_{i+1}))}{\vali(\zalloci)}\ge
1+a_i-x_i+\frac{\vali[i+1](R_{i+1}(r_{i+1}))}{\vali[i+1](\zalloci[i+1])}=1+a_i-x_i+x_{i+1},
\]
where the inequality follows from Claim~\ref{cl:greedy_reallocation}. Therefore, we can continue equation~\eqref{eq:nash_welfare_lm}: 
\be
\label{eq:logarithms_lemma}
\eqref{eq:nash_welfare_lm}\ge\sum_{i=1}^{k+1}\ln\left[1+a_i-x_i+x_{i+1}\right]-\sum_{i=1}^{k+1}\ln\left[1+a_i\right],
\ee
where $a_1=0$ and $x_1=x_{k+2}=0$. Moreover, we have $8\sqrt{\eps}\cdot\vali(\zalloci)\ge\vali(\alloci\setminus\zalloci)$ by the definition of $a_i$ for each $i\in[k]$ and $\zalloci[k+1]$ is the bundle that violates this inequality. Thus $a_i\bydef\frac{\vali(\alloci\setminus\zalloci)}{\vali(\zalloci)}\le8\sqrt{\eps}$ for every $i\in[k]$ and $a_{k+1}>8\sqrt{\eps}$. Note that we control each $x_i=\frac{\vali(R_i(r_i))}{\vali(\zalloci)}$ for $i\in\{2,...,k+1\}$ by choosing the corresponding number of items $r_i$ to be reallocated to agent $i-1$. Let $\Theta_i=\{\theta_{i,0}<...<\theta_{i,\ell_i}\}$ be the set of real numbers that $x_i$ can take for each $i\in\{2,...,k+1\}$. Then for each such $i$ we have

\begin{claim}
\label{cl:x_spaced}
The set $\Theta_i=\{\theta_{i,0}<...<\theta_{i,\ell_i}\}$ has $\theta_{i,0}=0$, $\theta_{i,\ell_i}=1$ and $|\theta_{i,j}-\theta_{i,j-1}|\le 2\eps$ for all $j$.
\end{claim}
\begin{proof}
Setting $r_i=0$ and $r_i=|\zalloci|$ yields $x_i=0$ and $x_i=1$, respectively. To bound the gap between consecutive $\theta_{i,j-1},\theta_{i,j}$ it suffices to bound $\frac{\vali(g)}{\vali(\zalloci)}\le 2\eps$ for any item $g\in\zalloci$. We have $\vali(\zalloci)\ge\frac{1}{2}\vali(\alloci)$ (by Lemma~\ref{lem:half_nsw_max}). Thus $\frac{\vali(g)}{\vali(\zalloci)}\le 2\frac{\vali(g)}{\vali(\alloci)}\le 2\eps$ by the large market assumption.
\end{proof}
Finally, we show that there are feasible $x_i$'s, so that the. RHS of \eqref{eq:logarithms_lemma} is strictly positive.
\begin{lemma}
\label{lem:sqrt_eps}
For any sequence of numbers $(a_i)_{i=1}^{k+1}$ such that $a_1=0, a_{k+1}\ge8\sqrt{\eps},$ and $a_i\in[0,8\sqrt{\eps}]$ there is a feasible solution $(x_i\in\Theta_i)_{i=1}^{k+2}$ with $x_1=x_{k+2}=0$ such that 
\[\text{RHS }\eqref{eq:logarithms_lemma}=
\sum_{i=1}^{k+1}\ln\left[1+a_i-x_i+x_{i+1}\right]-\sum_{i=1}^{k+1}\ln\left[1+a_i\right]>0
\]
\end{lemma}
\begin{proof} To prove the lemma, we use Karamata's inequality for the concave function $\phi(x)=\ln(1+x)$. Namely, we find a sequence of feasible $x_i$'s such that the ordered sequence $b_{(1)}\ge b_{(2)}\ge ... \ge b_{(k+1)}$ with $b_i=a_i-x_i+x_{i+1}$ for $i\in[k+1]$ is (strictly) majorized by the ordered sequence $a_{(1)}\ge a_{(2)}\ge ... \ge a_{(k+1)}$, i.e., $\sum_{i=1}^s b_i\le \sum_{i=1}^s a_i$ for any $1\le s\le k+1$, with equality for $s=k+1$. Then, $\sum_{i}^{k+1}\phi(b_i)>\sum_{i}^{k+1}\phi(a_i)$ for any strictly concave function $\phi$. 

First, we consider the case when there are consecutive $a_{i-1}$ and $,a_i$ such that $a_i-a_{i-1}>2\eps$. In this case, we set $x_j=0$ for all $j\neq i$ and $x_i=\theta_{i,1},$ i.e., the smallest non zero value in $\Theta_i$. Then almost all terms in both summations in the RHS~of~\eqref{eq:logarithms_lemma} are the same except for 
two terms $\ln(1+a_{i-1}+x_i)+\ln(1+a_i-x_i)$ and $\ln(1+a_{i-1})+\ln(1+a_i)$ in each summation. As $x_i\le 2\eps$, $b_{i-1}=a_{i-1}+x_i$ and $b_i=a_i-x_i$ are closer to each other than $a_{i-1},a_i$ with the same sum $b_{i-1}+b_i=a_{i-1}+a_{i}.$ Therefore, by Karamata's inequality $\ln(1+a_{i-1}+x_i)+\ln(1+a_i-x_i)-\ln(1+a_{i-1})-\ln(1+a_i)>0$ and the proof of the lemma is complete.

So, in the following, we assume that $a_i-a_{i-1}\le 2\eps$ for all $i\in\{2,...,k+1\}$. This means that $a_{i+1}\le 2\eps\cdot i$ and $a_{k+1-i}>8\sqrt{\eps}-2\eps\cdot i$ for any $i\in[k+1]$, as $a_1=0$ and $a_{k+1}>8\sqrt{\eps}$. Let $\ell\eqdef\left\lfloor\frac{2}{\sqrt{\eps}}\right\rfloor$, then $a_1,...,a_{\ell+1}\le 4\sqrt{\eps}\le a_{k-\ell+1},...,a_{k+1}$. First, we find an increasing  sequence of $(x_i)_{i=1}^{\ell+1}$, such that $x_1=0$, $x_{\ell+1}=1$ and each $b_i$ becomes closer to the median $4\sqrt{\eps}$ than $a_i$ (but still less than the median). To this end, we set $x_1=0$, and then, one after another, we set every following $x_{i+1}$ equal to $\argmax_{\theta\in\Theta_{i+1}}\left\{\theta\mid a_i+\theta-x_i\le 4\sqrt{\eps}\right\}$, where $i\in[\ell]$. Note that once the sequence of $x_i$'s reaches $1$ it stays equal to $1$, because maximal element $\bar{\theta}\in\Theta_{i+1}$ is $1$ and it satisfies the condition $a_i+\bar{\theta}-1\le 4\sqrt{\eps}$ for all $i\in[\ell]$. Now, let us verify that sequence $(x_i)_{i=1}^{\ell}$ reaches $1$. Assume to the contrary that $x_i<1$ for $i\le\ell$. That means by Claim~\ref{cl:x_spaced} that each $a_i-x_i+x_{i+1}$ must be at most $2\eps$ far from the median $4\sqrt{\eps}$, as otherwise we would increase $x_{i+1}$ and get closer to the median. Thus $a_i-x_i+x_{i+1}>4\sqrt{\eps}-2\eps$ for all $i\le\ell+1$. Hence,
\begin{multline*}
x_{\ell+1}=(x_2-x_1)+...+(x_{\ell+1}-x_{\ell})>\sum_{i=1}^{\ell}\left(4\sqrt{\eps}-2\eps-a_i\right)\ge
\sum_{i=1}^{\ell}\left(4\sqrt{\eps}-2\eps-2\eps(i-1)\right)\\
=\sqrt{\eps}\ell\cdot\left(4-\sqrt{\eps}(\ell+1)\right)
\ge\sqrt{\eps}\left(\frac{2}{\sqrt{\eps}}-1\right)\cdot\left(4-\sqrt{\eps}\left(\frac{2}{\sqrt{\eps}}+1\right)\right)
=\left(2-\sqrt{\eps}\right)^2\ge 1,
\end{multline*}
where in the second inequality we used $a_i\le 2\eps (i-1)$, and in the third inequality we used $\frac{2}{\sqrt{\eps}}\ge\ell=\left\lfloor\frac{2}{\sqrt{\eps}}\right\rfloor\ge\frac{2}{\sqrt{\eps}}-1$. We get $x_{\ell+1}> 1$, which contradicts $x_{\ell+1}\in\Theta_{\ell+1}$. Now, it is also easy to see that sequence $(x_i)_{i=1}^{\ell+1}$ is non-decreasing: either $x_i$ has already reached $1$ and stays constant, or $a_i-x_i+x_{i+1}=b_i\ge\frac{2}{\sqrt{\eps}}-2\eps\ge a_i$ for any $i\le\ell$. Analogously, we find a decreasing sequence\footnote{The argument mirrors the one for $i\in[\ell+1]$, i.e., we set $x_{k+2}=0$ first and then let each previous $x_{i-1}=\argmax_{\theta\in\Theta_{i-1}}\left\{\theta\mid a_{i-1}-\theta+x_i\ge 4\sqrt{\eps}\right\}$} of $(x_i)_{i=k-\ell+1}^{k+2}$, such that $x_{k-\ell+1}=1$, $x_{k+2}=0$ and each $b_i$ becomes closer to the median $4\sqrt{\eps}$ than $a_i$ (but still larger than $4\sqrt{\eps}$). We set the remaining $x_i$ for $\ell+1 < i< k-\ell+1$ to $1$, which gives us $b_i=a_i$ for $\ell+1 < i < k-\ell+1$ (recall that $x_{\ell+1}=x_{k-\ell+1}=1$). Now, it is easy to see that $\{b_i\}_{i=1}^{k+1}$ is majorized by $\{a_i\}_{i=1}^{k+1}$ which concludes the proof of the lemma by Karamata's inequality. 
\end{proof}
Lemma~\ref{lem:sqrt_eps} implies that Algorithm~\ref{alg:1} must terminate before any quantity  $(1+8\sqrt{\eps})\vali(\zalloci)$ becomes smaller than $\vali(\alloci)$, since otherwise there is an allocation $\hallocs$ with a higher Nash social welfare than the optimal one of $\allocs$. Moreover, for any edge $(i,\zalloci[j])$ in the EFX-feasibility-graph $\vali(\zalloci[j])\ge\vali(\zalloci)$. Hence, Algorithm~\ref{alg:1} for large markets outputs an EFX allocation $\yallocs$, with $(1+8\sqrt{\eps})\nashw(\yallocs)\ge(1+8\sqrt{\eps})\nashw(\zallocs)\ge\nashw(\allocs)=\optnw(\items)$. 
\end{proof}

\section{Discussion and Open Problems}
\label{sec:open}
We believe that our techniques could be used to show interesting interplays of more fairness notions with Nash welfare. For example, starting from the allocation computed by Algorithm 2, one could use the local-search algorithm of Lipton et al.~\cite{lipton2004approximately} to reallocate the removed items and get an EF1 allocation of {\em all} items. To the best of our knowledge, this is the first polynomial-time algorithm for computing a complete EF1 allocation that approximates maximum Nash welfare within a constant.

Still, the problem of whether EFX allocations of {\em all} items exist is widely open. If EFX allocations do not always exist, our proposed solution alleviates the existence issue while providing high efficiency guarantees. But we suspect that there is a monotonicity property that, if true, would not only show that EFX allocations always exist, but also that they are nearly-optimal in terms of Nash welfare. In particular, we suspect that adding an item to an allocation problem (that provably has an EFX allocation) yields another problem that also has an EFX allocation with at least as high Nash welfare as the initial one. Then, our Theorem~\ref{thm:main} would imply not only the existence of EFX allocations for all items, but also that the best among them is $2^{1-1/n}$-efficient. 
\subsection*{Acknowledgments}
Part of this work was done while authors IC and XH were visiting the Institute for Theoretical Computer Science at Shanghai University of Finance and Economics.

\bibliographystyle{ACM-Reference-Format}

\bibliography{ref}


\end{document}